\documentclass[copyright,creativecommons]{eptcs}
\providecommand{\event}{CL\&C 2014} 
\usepackage{breakurl}             
\usepackage{amsmath}
\usepackage{bussproofs}
\usepackage{stmaryrd}
\usepackage{cmll}
\usepackage{amssymb}
\usepackage{latexsym}
\usepackage{graphicx}
\usepackage{xcolor}
\usepackage{wasysym}\usepackage{hyperref}
\usepackage{amssymb}
\usepackage{amsthm}

\def\eg{e.g., }

\def\ie{i.e., }

\newcommand{\vv}{\langle}
\newcommand{\ww}{\rangle}

\newcommand{\NN}{\mathbb{N}}

\newcommand{\axi}{\mathsf{ax}}

\newcommand{\eqdef}{\stackrel{\text{\tiny \textnormal{DEF}}}{=}}
\newcommand{\eqrec}{\stackrel{\text{\tiny \textnormal{REC}}}{=}}

\newcommand{\st}{ \ \ \big{|} \ \ }

\newcommand{\roo}{\vv \, \ww}

\newcommand{\cV}{\mathcal{V}}

\newcommand{\cM}{\mathcal{M}}

\newcommand{\cP}{\mathcal{P}}

\newcommand{\cT}{\mathcal{T}}
\newcommand{\cU}{\mathcal{U}}
\newcommand{\ccW}{\mathcal{W}}

\newcommand{\Dis}[1]{\big( #1 \big)}

\newcommand{\arpa}[2]{\textstyle\vee_{#2} \bG_{#1}}
\newcommand{\arn}[2]{\textstyle\wedge_{#2} \bG_{#1}}
\newcommand{\arpp}[2]{\textstyle\vee_{#2} \bF_{#1}}
\newcommand{\arnn}[2]{\textstyle\wedge_{#2} \bF_{#1}}
\newcommand{\arppp}[2]{\textstyle\vee_{#2} \bH_{#1}}
\newcommand{\arnnn}[2]{\textstyle\wedge_{#2} \bH_{#1}}

\newcommand{\TRUE}{\pmb{t}}
\newcommand{\FALSE}{\pmb{f}}

\newcommand{\bu}{\mathbf{u}}

\newcommand{\bw}{\mathbf{w}}

\newcommand{\bL}{\mathbf{L}}

\newcommand{\bF}{\mathbf{F}}
\newcommand{\bG}{\mathbf{G}}
\newcommand{\bH}{\mathbf{H}}
\newcommand{\bS}{\mathbf{S}}
\newcommand{\bT}{\mathbf{T}}
\newcommand{\bU}{\mathbf{U}}
\newcommand{\bW}{\mathbf{W}}
\newcommand{\bV}{\mathbf{V}}
\newcommand{\bE}{\mathbf{E}}

\newcommand{\bba}{\mathbf{a}}

\newcommand{\bbv}{\mathbf{v}}

\newcommand{\len}{\textnormal{lg}}

\newcommand{\dom}{\ensuremath{\textnormal{dom}}}

\newcommand{\seq}{\textnormal{SEQ}}
\newcommand{\conf}{\textnormal{CONF}}
\newcommand{\rules}{\textnormal{RULES}}

\newcommand{\TES}{\textnormal{TESTS}}
\newcommand{\envir}{\textnormal{ENV}}
\newcommand{\CT}{\textnormal{IT}}

\newcommand{\sR}{\mathsf{R}}

\newcommand{\sL}{\mathsf{L}}

\newcommand{\arr}{\longrightarrow}

\newtheorem{theorem}{Theorem}[section]   
\newtheorem{lemma}[theorem]{Lemma}
\newtheorem{proposition}[theorem]{Proposition}

\theoremstyle{definition}

\newtheorem{Definition}[theorem]{Definition}
\newtheorem{remark}[theorem]{Remark}

\newtheorem{Example}[theorem]{Example}

\newtheorem{Notation}[theorem]{Notation}

\newtheorem{notation}[theorem]{Notation and Terminology}{\bfseries}{}

\newcommand{\squishlist}{
 \begin{list}{$\bullet$}
  { \setlength{\itemsep}{0pt}
     \setlength{\parsep}{3pt}
     \setlength{\topsep}{3pt}
     \setlength{\partopsep}{0pt}
     \setlength{\leftmargin}{1em}
     \setlength{\labelwidth}{1.5em}
     \setlength{\labelsep}{0.5em} } }
\newcommand{\squishend}{
  \end{list}  }

\newcommand{\squishlistt}{
 \begin{list}{$\bullet$}
  { \setlength{\itemsep}{0pt}
     \setlength{\parsep}{3pt}
     \setlength{\topsep}{3pt}
     \setlength{\partopsep}{0pt}
     \setlength{\leftmargin}{1em}
     \setlength{\labelwidth}{1.5em}
     \setlength{\labelsep}{0.3em} } }
\newcommand{\squishendd}{
  \end{list}  }

\newcommand{\sqi}{
 \begin{list}{$\bullet$}
  { \setlength{\itemsep}{0pt}
     \setlength{\parsep}{3pt}
     \setlength{\topsep}{3pt}
     \setlength{\partopsep}{0pt}
     \setlength{\leftmargin}{1.4em}
     \setlength{\labelwidth}{1.5em}
     \setlength{\labelsep}{0.3em} } }
\newcommand{\sqe}{
  \end{list}  }

\title{Infinitary Classical Logic: \\  Recursive Equations and  Interactive Semantics}
\author{Michele Basaldella\thanks{Supported by the ANR project ANR-2010-BLAN-021301 LOGOI.}
\institute{Universit\'e d'Aix--Marseille, CNRS, I2M, Marseille, France}
\email{michele.basaldella@gmail.com}
}
\def\titlerunning{Infinitary Classical Logic:   Recursive Equations and  Interactive Semantics}
\def\authorrunning{Michele Basaldella}
\begin{document}
\maketitle

\begin{abstract}
In  this paper, we  present an interactive semantics for derivations  in an infinitary
extension of classical logic.
The formulas of our language are possibly infinitary
 trees labeled by propositional variables and logical
connectives.
We show  that in our setting
 every recursive formula equation
has a unique solution.
As for derivations, we use an infinitary variant of
 Tait--calculus to derive sequents.

 The interactive  semantics for derivations that we introduce  in this article  is  presented as a  debate (interaction tree)
between a  test $\cT$ (derivation candidate, Proponent) and
an environment   $\neg \bS$ (negation of a sequent, Opponent).
We   show a
completeness theorem for derivations that we call
interactive completeness
 theorem: the interaction between $\cT$ (test) and
 $\neg \bS$ (environment)
does not produce errors   (\ie Proponent wins)  just in case $\cT$ comes from a syntactical derivation
of  $\bS$.

\end{abstract}

\section{Introduction}

In this article, we present an \emph{interactive semantics}
for \emph{derivations} | \ie formal proofs | in a proof--system
that we call \emph{infinitary classical logic}. \\

\vspace{-0.35cm}

 \noindent \textbf{Infinitary classical logic.} \
The system we consider is an infinitary extension of
\emph{Tait--calculus} \cite{Tait68},  a  sequent calculus  for \emph{classical logic}
which is often used to analyze the proof theory of classical arithmetic and its fragments.
In Tait--calculus, formulas
are  built from positive and negated propositional variables by using
disjunctions $\vee$ and conjunctions $\wedge$
 of arbitrary (possibly infinite)
arity. Negation $\neg$ is  defined
by using a generalized form of  De Morgan's laws.
As for derivations,
sequents of formulas are derived by means of rules of
inference with a possibly infinite number of premises.
In Tait--calculus, formulas and derivations | when seen as trees |
 while not necessarily finitarily branching, are \emph{well--founded}.
In this work, we remove the assumption
of well--foundedness, and we let
formulas and derivations be
 infinitary in a broader sense. What  we get is the system that we
 call \emph{infinitary classical logic}. \\

 \vspace{-0.35cm}

 \noindent \textbf{Recursive equations.} \
 The main  reason for introducing infinitary classical logic
 is our interest in studying \emph{recursive}
 (\emph{formula}) \emph{equations} in a classical context.
 Roughly speaking, by recursive  equation we  mean  a
pair of formulas $(\bbv,\bF)$,  that we write as  $\bbv \eqrec \bF$, where
$\bbv$ is an atom and  $\bF$ is a formula which depends on $\bbv$, \ie such that $\bbv$ occurs in $\bF$. For instance,
$\bbv \eqrec \neg \bbv \vee \bbv$ is a recursive equation.
  A \emph{solution} of a recursive equation, say $\bbv \eqrec \neg \bbv \vee \bbv$,
 is any formula $\bG$ which is equal   to $\neg \bG \vee \bG$.
  Solutions of recursive equations are often called
 \emph{recursive types},
 and they have been studied extensively  in the literature
 (see \eg \cite{coppo,MelVou} and the references therein).
In this area,  one usually aims at finding a   mathematical space  in which the (interpretations of) equations
have a unique solution (or, at least, a canonical one).
In this paper,
we \emph{define} formulas as (possibly infinitary) labeled trees,
and we prove  existence and uniqueness of  solutions of equations
within the ``space" of the formulas.
As it turns out, if $\bG$ is the solution of a recursive equation (\ie a recursive type), then $\bG$
is not  well--founded. This
fact motivates us to consider  infinitary formulas in our broader sense. \\
\vspace{-0.35cm}

 \noindent \textbf{Derivations and tests.} \
 Since formulas are infinitary, we  let derivations  be infinitary as well.
We obtain a cut--free sequent
calculus in which the
solutions of all recursive equations
are derivable (in the sense of Theorem \ref{recder}(2)).  As expected, since we deal with
ill--founded (\ie non--well--founded) derivations,
 the price to pay for this huge amount of expressivity is that our calculus
is inconsistent (in the  sense of Remark \ref{remcons}(c)).
In spite of this,  it is precisely
this notion of infinitary  derivation that we want
to study in this work.  To this aim, we introduce a semantics \emph{for derivations},  as we now explain.

Traditionally, the proof theory of classical logic is centered around the notion of  \emph{derivability}.
In this paper, we are interested in analyzing
the structure of  our infinitary \emph{derivations}. To this aim, we introduce the notion of \emph{test}.  A {test} $\cT$ is a
 tree labeled by logical rules
(no sequents), and the fundamental relation between tests and derivations can be informally stated as follows: given a sequent $\bS$,
 if it is the case that by adding sequent information in an appropriate way to $\cT$  we obtain  a  derivation of  $\bS$, then
 we say that
  \emph{$\cT$ comes from a derivation
of   \/$\bS$}.
In this article,  the  \emph{syntactical} concept that we  investigate is ``$\cT$ comes from a derivation
of \/$\bS$"  rather than the traditional one ``$\bS$ is derivable."
Also, note that our concept is \emph{stronger} than the usual  one:
 if $\cT$ comes from a derivation of $\bS$, then $\bS$ is obviously derivable!

To grasp the idea behind our  concept from another  viewpoint, consider the lambda calculus. By the Curry--Howard correspondence, untyped  lambda terms can be seen as ``tests"
for  natural deduction derivations, and  ``$\cT$ comes from a derivation
of \/$\bS$"
 can be read  as ``the untyped lambda term $\cT$
has (simple) type  $\bS$ in the Curry--style type assignment."\\
\vspace{-0.35cm}

 \noindent \textbf{Interactive semantics and completeness.} \
Traditionally, in order to study   the concept of derivability,   one  introduces a notion of model and eventually shows a \emph{completeness
theorem}:  the \emph{syntactical} notion of \emph{derivability}
and the \emph{semantical} notion  of \emph{validity}
(as usual, \emph{valid} means ``true in every model")
coincide. Here, we are interested in
proving a completeness theorem as well.
But, since we   replace the syntactical concept of derivability
with ``$\cT$ comes from a derivation
of \/$\bS$", we   need to replace the semantical notion  of  validity with something else.
So, we now let our interactive semantics enter the stage.

In few words, our interactive semantics is organized as follows: first, we introduce the notion
of \emph{environment} $\neg \bS$ (the  negation  of a  sequent $\bS$) and then,
we make the test $\cT$ and the environment $\neg \bS$ \emph{interact}.
More precisely, we  introduce the notion of \emph{configuration} (a pair of the form $(\cT,\neg \bS)$) and define a procedure which makes configurations evolve from the initial configuration $(\cT,\neg \bS)$ by means of
a  transition relation. As a result, we get a tree of configurations
that we call an \emph{interaction tree}.
The procedure which determines
the interaction tree is \emph{interactive} in the sense
that it can be seen as a debate
between two players: Proponent  (the test $\cT$)
and   Opponent (the environment $\neg \bS$).
Opponent  \emph{wins the debate} if  the interaction between $\cT$ and $\neg \bS$ produces
 an
\emph{error}, \ie a position
in the interaction tree generated  from  $(\cT,\neg \bS)$ which is labeled by an  error symbol. Otherwise, Proponent \emph{wins the debate}. Our main result is the \emph{interactive completeness theorem}:
\vspace{-0.1cm}
\squishlist
\item[] {\centering
$\cT$ comes from a derivation  of $\bS$ \ \ \ \ \ \ if and only if \ \ \ \ \ \ $\left.
  \begin{array}{l}
  \hbox{the interaction between $\cT$ and $\neg \bS$} \\
  \hbox{does not
produce errors.}   \end{array}
\right.$ \hfill ($\star$) \par}
\squishend

\noindent That is, $\cT$ comes from a derivation  of $\bS$ if and only if Proponent wins the debate.

The motivation for introducing our semantics comes
from our interest   in extending
the completeness theorem of ludics \cite{locus} to logics which
are not necessarily polarized fragments of linear logic.
 The completeness theorem of \cite{BT}
(called interactive completeness also there) can be stated, up to terminology and notation, in our setting as
\vspace{-0.1cm}
\squishlist
\item[] {\centering
$\cT$ comes from a derivation  of $\bS$ \ \ \ \ \ \ if and only if \ \ \ \ \ \ $\left.
  \begin{array}{l}
  \hbox{for  no $\cM \in \neg \bS$, the interaction between} \\
  \hbox{$\cT$ and $\cM$
produces errors.}   \end{array}
\right.$ \hfill ($\star\star$) \par}
\squishend
The crucial point is, of course,
the RHS of ($\star\star$).
Here, $\bS$ represents a formula of a polarized fragment of linear logic, and $\cT$ and $\cM$ are designs
(proof--like objects similar to our tests).
In \cite{BT}, formulas
are interpreted as sets of designs such that  $\neg \neg \bS = \bS$, where $ \neg \bS$ is the set
of designs given by:  $\cM \in \neg \bS$
just in case for every $\cP \in \bS$, the interaction between $\cP$  and $\cM$
does not produce errors (so, the RHS of ($\star\star$) means $\cT \in \neg \neg \bS = \bS$).
In \cite{BT},  the result of the interaction between $\cT$ and $\cM$   is determined by a procedure of reduction for designs which reflects the procedure of cut--elimination of the underlying logic.
Of course, it would be very nice if we could
use the same approach  in our setting!

Unfortunately, if one tries
to remove the polarities from  ludics, then one encounters
several \emph{technical} problems
related to cut--elimination (that we do not discuss here).
To the present author, the
most convenient way to cope with non--polarized logics, is to build
a new framework from the very beginning, keeping the \emph{format} of the statement ($\star\star$) as guiding principle, the rest being | in case |  sacrificed.
The  choice here is purely personal:
the present author believes that
the significance of ludics is ultimately justified by the  completeness theorem not \eg
by the fact that   the interpretation of the logical formulas is induced by a procedure of reduction for designs.
This describes the origin of this work.

Finally, we note that
the RHS of ($\star$)
is just the  RHS of ($\star\star$) in case $\neg \bS$ is a singleton
(and so, it can be identified with
$\neg \bS$ itself). Indeed, this is the case for the interactive semantics presented in this
paper.

As for future work, we plan to adapt our interactive  semantics to analyze derivations in
second--order propositional classical logic.\\

\vspace{-0.35cm}

 \noindent \textbf{Outline.} \  This paper is organized as follows.
In Section \ref{sec2} we recall some preliminary
notions about labeled trees.
In Section \ref{sec3} we introduce
formulas and
 recursive formula equations, 	and
we prove the existence and the uniqueness of solutions of equations.
We also define the  notion of derivation and discuss
the derivability of some sequents.
In Section \ref{sec4} we present our interactive semantics
and prove the   interactive  completeness theorem.

\section{Preliminaries: Positions and Labeled Trees} \label{sec2}

In this section, we recall  the basic  notions of
position
and labeled tree.
We also establish some  notation and terminology that
we  extensively  use in the sequel.

\begin{Definition}[Position, length]
Let $\NN$ be the set of the natural numbers.
Let $\infty$ be any object such that
$\infty \notin \NN$. We call \textbf{position}
any function
$p : \NN \arr \NN \, \cup \, \{\infty\}$
which satisfies the following property:
\squishlist
\item[$\phantom{ab}$ (P)] \quad for some    $n \in \NN$, \
$p(k) \in \NN$   for  all  $ k < n$
\ and \ $p(k) = \infty$  for  all $ k \geq n$.
\squishend
Since $\infty \notin \NN$,  the natural number $n$  above
is unique. We call it the \textbf{length} of the position and  denote it by $\len(p)$.
The set of all  positions
is denoted by  $\NN^\star$ and we
use $p,q,r\ldots$ to range over its elements.
\hfill $\triangle$
\end{Definition}

\begin{notation} Let $p$, $ q$ and $r$ be positions,
and let $U$, $V$ and $W$ be sets of positions.

(1) A position  $p$ is also written  as $\vv p(0),\ldots,p(\len(p)-1)\ww$.
In particular, we write $\roo$ for the unique position of length $0$ that we call the \emph{empty position}.

(2)
The \emph{concatenation}
of  $p$ and $q$
is the position $p \star q$ defined as follows:

\squishlist
\item[] {\centering
$p \star q \, (k) \ \eqdef \ $ $\left\{
  \begin{array}{ll}
    p(k) \enspace&  \hbox{if $ k < \len(p)$} \\
    q(k - \len(p)) \enspace &  \hbox{if $  k \geq \len(p)$\enspace, \qquad for $k \in \NN$\enspace.}   \end{array}
\right.$ \par}
\squishend

\noindent In other words,
$p \star q \ = \vv p(0),\ldots,p(\len(p)-1), q(0),\ldots,q(\len(q)-1)\ww$.
The operation of concatenation
is associative (\ie $(p \star q) \star r = p \star (q \star r) $)
and it has $\roo$ as neutral element (\ie $\roo \star p = p = p \star \roo$). Note also that $\len(p \star q) = \len(p) + \len(q)$.

(3) We write $p \sqsubseteq q$     if
$q = p \star t$  for some $t \in \NN^\star$, and we
 say that $q$ is an \emph{extension of} $p$
and that $p$ is a \emph{restriction of} $q$.
If $ \len(t) \geq 1$ (resp. $ \len(t) = 1$), then we also say that $q$ is a \emph{proper}
(resp. \emph{immediate}) extension of $p$
and that $p$ is a  \emph{proper} (resp. \emph{immediate})  restriction of $q$, and we write $p \sqsubset q$ (resp. $p \sqsubset_1 q$).


(4) We write $U\star V$
for the set
$\{ p \star q  \st  p \in U $  and $q \in V\}$.
Note that $U \star (V \star W) = (U \star  V) \star W $,
and that $ U \star \bigcup_{ i \in I} V_i = \bigcup_{i \in I} (U \star V_i)$, for every family $\{V_i\}_{i \in I}$ of sets
of positions.
 \hfill $\triangle$
\end{notation}

\begin{Definition}[Labeled tree, domain] \label{tree}  Let $L$ be a set.
Let $\infty_L$ be any object such that
$\infty_L \notin L$.
A \textbf{tree}  \textbf{labeled by} $L$ is a function
$T : \NN^\star \arr L \, \cup \, \{ \infty_L \}$
which satisfies the following properties:
\squishlist
\item[$\phantom{ab}$ (T$_1$)] $T(\roo) \in L$;
\item[$\phantom{ab}$ (T$_2$)]  if $T(p) \in L$ and $q \sqsubseteq p$,  then $T(q) \in L$.
\squishend
We call  the set $\{ p \in \NN^\star \st T(p) \in L \}$ the \textbf{domain}
of $T$ and we denote it as $\dom(T)$. \hfill $\triangle$
\end{Definition}

Let $T$ and $U$ be  trees labeled by $L$.
Since  $T(q) = \infty_L$ for all
$q \notin \dom(T)$,  the tree $T$ is completely determined
by  the values it takes on its domain.
In particular,
\squishlist
\item[] \centering{$T = U$  \ \ \ \ if and only  if  \ \ \ \
$\dom(T) = \dom(U)$ \,   and \,  $T(p) = U(p)$ for all $p \in \dom(T)$\enspace. \par}
\squishend

\begin{notation} Let $T$ and $U$ be  trees labeled by $L$, and let  $p \in \dom(T)$.

(1)
 We say that
$p$ is a \emph{leaf of $T$} if  there is no $q \in \dom(T)$ such that $p \sqsubset q$.

(2)
The \emph{subtree of $T$ above $p$} is the tree $T_p$ labeled
by $L$ defined as follows:

\squishlist
\item[] {\centering \hspace{-0.8cm}
$\dom(T_p) \ \eqdef \  \{ q \in \NN^\star \st p \star q \in \dom(T) \}$ \ \ \ \ \ \  and \ \ \ \ \ \
  $T_p(q) \ \eqdef \  T(p \star q)$\enspace, \  for
 $q \in \dom(T_p)$\enspace.  \par}
\squishend
Note that we have $T_{\roo} = T$  and $(T_p)_q = T_{p \star q}$, for every  $p$ and $q$ in $\NN^\star$  such that $p \star q \in \dom(T)$.

We  say that $U$ is a \emph{subtree of} $T$
if $U = T_q$, for some $q \in \dom(T)$.
If  $\len(q) = 1$, then we also say that $U$ is
 an \emph{immediate subtree of $T$}.

(3)
 If $L$ is a product, \ie $L = A \times B$ for some sets $A$ and $B$, then
 we write $T_{\sL}(p)$ and $T_{\sR}(p)$
 for the  left and the right component
 of $T(p)$ respectively (\ie if
$T(p) = (a,b)$,  then
$T_{\sL}(p) = a $ and $T_{\sR}(p) = b$).

(4) We say that $T$ is  \emph{ill--founded}
if there exists a function $f : \NN \arr \dom(T)$
such that
$f(n) \sqsubset f(n+1)$ for every $n \in \NN$. We also say that $T$ is \emph{well--founded} if it is not ill--founded.  \hfill $\triangle$
\end{notation}

\section{Infinitary Classical Logic} \label{sec3}

In this section, we present
 our infinitary version of classical logic.
In Subsection \ref{forsec}
we introduce formulas
as possibly infinitary labeled trees,
and in Subsection \ref{receqsec}
we introduce the notion of recursive formula equation and prove that
solutions of equations exist and they are unique.
Finally, in Subsection \ref{subder} we introduce the concept of derivation
and observe some basic properties.


\subsection{Formulas} \label{forsec}
We  now define the concept of   formula
and the operation of  negation.

 \begin{Definition}[Propositional variable]
 Let $+$ and $-$ be two distinct symbols. Let $\cV \eqdef \{+\} \times \NN$ and $\neg \cV \eqdef  \{-\} \times \NN$.
  We call  the elements of $\cV$ (resp. $\neg \cV$)  \textbf{positive} (resp. \textbf{negated}) \textbf{propositional
 variables}, and
 we denote a generic element $(+,v)$ of $\cV$ (resp. $(-,v)$ of  $\neg \cV$) by $\bbv$  (resp. $\neg \bbv$).
   \hfill$\triangle$\end{Definition}

\begin{Definition}[Formula] \label{formula}
 Let $\vee$ and $\wedge$ be two distinct symbols
 such that $ \{ \vee ,\wedge \} \, \cap \, \big(\cV \cup \neg \cV \big) = \emptyset$.
We call  \textbf{formula}  any tree  $T$ labeled by $\{ \vee ,\wedge \} \, \cup \,  \cV \cup \neg \cV$ which satisfies the following property:
 \squishlist
\item[$\phantom{ab}$ (F)]  \qquad   for every $p \in \dom(T)$, if   $T(p) \in \cV \cup \neg \cV$,  then  $p$ is a leaf of $T$.
\squishend
In the sequel, we use the letters $\bF, \bG ,\bH ,\ldots$ to range over formulas.
\hfill$\triangle$
 \end{Definition}

\begin{notation}
Let $\bF$ be a formula.

(1) Since
for every $p \in \dom(\bF)$, the  subtree
$\bF_p$ is a formula, we say that
$\bF_p$ is a \emph{subformula of} $\bF$ and that $\bF_p$ \emph{occurs in} $\bF$. If
$\len(p) =1$ then we also say that
$\bF_p$ is an \emph{immediate subformula of} $\bF$.

(2) If  $\bF(\roo)  \in \cV \cup \neg \cV$,  then $\dom(\bF) = \{\roo\}$
and we say that $\bF$ is an \textbf{atom}.
If  $\bF(\roo) = \bbv \in \cV$ (resp. $\bF(\roo) = \neg \bbv \in \neg \cV$),  then we abusively   denote $\bF$ by $\bbv$ (resp. $\neg \bbv$) and we  say that $\bF$ is  a \textbf{positive} (resp. \textbf{negated}) \textbf{atom}.

(3) If ${\bF}(\roo)  \in \{ \vee, \wedge\}$, then we call $\bF$  a \textbf{compound formula}. The set of natural numbers
\squishlist
\item[] \centering{$I \ \eqdef \ \{i \in \NN \st \vv i \ww \in \dom(\bF) \}$ \par}
\squishend
\noindent is said to be the \textbf{arity of}  $\bF$.
Let $\bF(\roo) = \vee$ (resp. $\bF(\roo) = \wedge$). Since for every $i \in I$ the formula $\bF_{\vv i \ww}$ is
an immediate subformula of $\bF$, we also denote $\bF$ by
  $ \vee_I \bF_{\vv i \ww} $
(resp.  $\arnn{\vv i \ww}{I}$)
and we say that  $\bF$ is a
\textbf{disjunctive} (resp. \textbf{conjunctive}) \textbf{formula}.
If $I = \{0,\ldots,n-1\}$ for some $n \in \NN$,
 then
we also write $\vee [ \bF_{\vv 0 \ww},\ldots, \bF_{\vv n-1 \ww} ]$
(resp. $\wedge [ \bF_{\vv 0 \ww},\ldots, \bF_{\vv n-1 \ww} ]$)
for $\bF$.
Obviously, if we write  something like ``let $\bF$ be  $
\vee [ \bG_0,\ldots, \bG_{n-1} ]$ \ldots " (resp.
$\wedge [ \bG_0,\ldots, \bG_{n-1} ]$), then
we mean that $\bF$ is a disjunctive (resp. conjunctive)
 formula of arity $\{0,\ldots, n-1\}$ and that $\bG_k = \bF_{\vv k \ww}$, for each $k < n$.
Similarly, we may use the expression $\vee[\bG,\bH]$ (resp. $\wedge[\bG,\bH]$) to denote the disjunctive (resp. conjunctive) formula
$\bF$ whose arity is $\{0,1\}$ and such that $\bF_{\vv 0 \ww} = \bG$ and
 $\bF_{\vv 1 \ww} = \bH$, and so on.
 Finally, if
$I  = \emptyset$ | that is, $\dom(\bF) = \{\roo\}$) |
then we write $\FALSE$ (for \emph{false}) and  $\TRUE$ (for \emph{true})
rather than $\vee [ \ ]$ (resp. $\wedge [ \ ]$) respectively.
 \hfill$\triangle$

\end{notation}

 \begin{Definition}[Negation]
 Let $\bF$ be a formula.
 The formula $\neg\bF$, that we call
 the \textbf{negation of} $\bF$, is defined
 as follows:
$\dom(\neg\bF) \eqdef  \dom(\bF)$ and

\squishlist
\item[]
{\centering
${\neg\bF}(p) \ \eqdef \ $ $\left\{
  \begin{array}{rl}
    \neg \bbv\enspace &  \hbox{ if ${\bF}(p) = \bbv$} \\
    \bbv\enspace &  \hbox{ if ${\bF}(p) = \neg \bbv$} \\
    \vee\enspace &  \hbox{ if ${\bF}(p) = \wedge$} \\
    \wedge\enspace &  \hbox{ if ${\bF}(p) = \vee$\enspace, \qquad for $p \in \dom(\neg\bF)$\enspace. }
  \end{array}
\right.$

\vspace{-0.55cm}
\hfill$\triangle$ \par}
\squishend

 \end{Definition}
 We observe that the negation is involutive, \ie
$\neg \neg \bF = \bF$, for every formula $\bF$. Furthermore,
\squishlist
\item[] \centering{$\neg\arpp{\vv i \ww}{I} \ = \  \wedge_I \neg\bF_{\vv i \ww}$ \ \ \ \ and \ \ \ \
$\neg\arnn{\vv i \ww}{I} \ = \  \vee_I \neg\bF_{\vv i \ww}$\enspace.  \par}
\squishend

 According to Definition \ref{formula},
 formulas are allowed to  be  \emph{infinitary}:
they  may have an infinite set $I \subseteq \NN$ as  arity, and they  can also be ill--founded. In
Tait's work \cite{Tait68},  the situation is rather different:
 the arity of a compound  formula
 need not  be a subset of $\NN$, and  only
well--founded formulas are  considered.

Let us  discuss our choices.
As for the first difference,
we restrict  attention to sets of natural numbers mainly for expository reasons; a pleasant  consequence of
this choice is
 that we can define  sequents as \emph{formulas}  (Definition \ref{seq})
rather than finite sets (as  in \cite{Tait68}),
multi--sets, or sequences of  formulas.
As for the second one, we have to  consider ill--founded
formulas because   solutions of   recursive equations
are formulas which \emph{always}   have this property (see Theorem \ref{recthem}).

We finally observe that
the principles of (transfinite)
induction and recursion cannot be applied to
ill--founded
formulas.  In particular, | even though we are in classical world | it is not possible to give to our formulas
a Tarskian definition of \emph{truth}. Nevertheless, in our setting it is possible to
define a reasonable notion of \emph{derivation}, as we do  in
 Subsection  \ref{subder}.

 \subsection{Recursive Formula Equations} \label{receqsec}
In this subsection, we  define the concepts of recursive equation and solution of an equation. We
 prove  that every recursive equation
has a unique solution. We also give some concrete examples.

\begin{Definition}[Substitution] \label{subdef}
Let $\bF$ and $\bG$ be  formulas, and let $\bbv$ be a
\emph{positive} atom.
We define the formula $\bF[\bG/\bbv]$
obtained by the \textbf{substitution of $\bG$ for $\bbv$ and
of $\neg\bG$ for $\neg \bbv$ in $\bF$} as follows.
Let $ R \eqdef \{ r \in \dom(\bF) \st {\bF}(r) \in \{\bbv, \neg \bbv\}  \}$
and $S \eqdef \dom(\bF) \setminus R$.
We set $\dom(\bF[\bG/\bbv]) \ \eqdef \ S \, \cup \, \big(R \star  \dom(\bG) \big)$ and

\squishlist
\item[]
{\centering
$\bF[\bG/\bbv](p) \ \eqdef \ $ $\left\{
  \begin{array}{rl}
  {\bF}(p)\enspace &  \hbox{if  $p \in S$} \\
     {\bG}(q)\enspace &  \hbox{if  $p = r \star q$ and   $\bF(r) = \bbv$} \\
  \neg{\bG}(q)\enspace &  \hbox{if  $p = r \star q$ and   $\bF(r) = \neg \bbv$\enspace, \ \  for  $p \in \dom(\bF[\bG/\bbv])$\enspace. }   \end{array}
\right.$

\vspace{-0.53cm}
\hfill$\triangle$ \par}
\squishend
\end{Definition}
The correctness of the previous definition is justified by the following lemma.

\begin{lemma} \label{lemsub}
Let  \/$R$\@ and \/$S$\@ be as in \emph{Definition \ref{subdef}}.
\sqi
\item[\emph{(a)}]
For every \/$r  \in R$\@ and every  \/$ t \in \NN^\star$\@, if \/$p= r \star t$\@  then  \/$p \notin S$\@.
\item[\emph{(b)}]  Let  \/$V \subseteq \NN^\star$\@, and let   \/$p \in R \star V$\@.
Suppose that  there are  \/$r$\@,  \/$r'$ in \/$R$\@ and    \/$q$\@,   \/$q'$ in  \/$V$\@ such that  \/$ p =   r \star q$\@ and \/$p = r' \star q'$\@.
Then, \/$r = r'$\@ and \/$q = q'$\@.

\sqe
\end{lemma}
\begin{proof}
(a)
If $t = \roo$, then $p = r \star \roo  = r \in R$.
If $ t \neq \roo$, then $ p \notin \dom(\bF)$,  as $r$ is a leaf of $\bF$. Hence, $ p \notin S$.

(b)   Since   both $r$ and $r'$
are restrictions of $p$,  we have
$r \sqsubseteq r'$ or $r' \sqsubseteq r$.
 Since  $r$ and $r'$ are leaves of $\bF$,
we conclude  $r = r'$ (and  hence $q = q'$).  \end{proof}
The following proposition easily follows from our definition of substitution.

\begin{proposition} \label{subst} Let  \/$\bF$\@ and \/$\bG$\@ be  formulas, and
let \/$\bbv$\@  be a positive atom. Then:
\squishlist
\item[$\phantom{ab}$ \emph{(1)}]
$\bbv[\bG/\bbv] = \bG$\@ and  \/$(\neg\bbv)[\bG/\bbv] = \neg\bG$;
\item[$\phantom{ab}$ \emph{(2)}] $\big(\arpp{\vv i\ww}{I}\big)[\bG/\bbv] = \vee_{I} \big(\bF_{\vv i \ww}[\bG/\bbv]\big)$\@ and
\/$\big(\arnn{\vv i\ww}{I}\big)[\bG/\bbv] = \wedge_{I} \big(\bF_{\vv i \ww}[\bG/\bbv]\big)$\@;
\item[$\phantom{ab}$ \emph{(3)}] $\bF[\bG/\bbv] = \bF$\@,  if neither \/$\bbv$\@ nor \/$\neg\bbv$\@ is a subformula of  \/$\bF$\@;
\item[$\phantom{ab}$ \emph{(4)}]
$\neg(\bF[\bG/\bbv]) = (\neg\bF)[\bG/\bbv]$\@. \hfill $\square$
\squishend
\end{proposition}

\begin{Definition}[Recursive formula equation, solution] \label{receq}
A \textbf{recursive formula equation} (or  \textbf{recursive  equation}, or just \textbf{equation}) is an ordered  pair of formulas
$(\bbv, \bF)$, that we write as $ \bbv \eqrec \bF$,  such that: \squishlist
\item[$\phantom{ab}$ (R$_1$)] $\bbv$ is a positive atom;
\item[$\phantom{ab}$ (R$_2$)]  $\bF$ is  a compound formula;
\item[$\phantom{ab}$ (R$_3$)] $\bF(p) \in \{\bbv, \neg \bbv\}$,  for some $p \in \dom(\bF)$.
\squishend
 A \textbf{solution} of  $\bbv \eqrec \bF$ is a formula $\bG$ such that  $\bG = \bF[\bG/\bbv]$.
\hfill $\triangle$
\end{Definition}

We now discuss the previous definition.
To begin with, note that by (R$_2$),    a pair of atoms such as
$(\bbv, \bba)$ is   \emph{not} a recursive equation.
The reason to exclude such pairs is
to avoid  to consider $(\bbv, \bbv)$ | which would be a trivial equation, as every formula would be a solution |
and $(\bbv, \neg \bbv)$ |
which would have  no solution at all.
Up to now, our  choices are perfectly in line with \cite{Courcelle,coppo}.
In our setting,  we also impose the additional condition   (R$_3$). The aim of this clause is to exclude pairs
of the form $(\bbv, \bF)$ where neither
$\bbv$ nor $\neg \bbv$ actually occurs in $\bF$.
The reason is  that   pairs like
the one above would be trivial equations as well:
by Proposition \ref{subst}(3), $\bF$ itself would  be
 the unique solution, as   we have $\bF[\bG/\bbv] = \bF$ for every formula $\bG$.

We now turn attention to the literature on recursive types.
In this topic, the theorem which states
 existence and  uniqueness of solutions of recursive equations
is perhaps the most important result.
One usual way to prove  it is to show that  the mathematical space (in our case, it would be the set of formulas) forms a complete metric space with
respect to some  metric.  Then, the result  follows
by applying Banach's fixpoint theorem,  by using the fact
that the operation of substitution induces
a contractive map from the space  to itself  \cite{Courcelle,coppo}.
Another traditional way
 to prove that result requires to introduce  in the space  some  notion of approximation.  One then shows that   suitable  sequences of ``lower" (resp. ``upper") approximations
converge to a ``lower" (resp. ``upper") solution of the equation.
To prove the result, one eventually proves    that the two solutions coincide (see \eg \cite{MelVou,BT} and the references therein).

By contrast, in our setting we are able to prove the result in  a  direct and elementary way; our method does not require to  explicitly introduce any sort of metric or notions of approximation.
  We do  not claim that our \emph{result} is new, as there are several \emph{similar}  results in the literature.
However, we believe that
our \emph{proof} deserves some attention, as  it is quite simple and self--contained.

\begin{theorem}[Existence and uniqueness of solutions] \label{recthem}
Every recursive equation  \/$\bbv \eqrec \bF$\@  has
a unique solution \/$\bG$\@. Furthermore,
\/$\bG$\@ is ill--founded.  \end{theorem}
\begin{proof}
Let $R$ and $S$ be as in Definition \ref{subdef}.
First, we observe that by  condition   (R$_2$)  of Definition \ref{receq}, $\roo \in S$.
In particular, $\roo \notin R$. Furthermore, by condition (R$_3$), $R$ is non--empty.

We now remark that by definition of substitution, if  $\bH$
is a solution
of the equation $\bbv \eqrec \bF$,
then the set of positions $\dom(\bH)$ has to satisfy $\dom(\bH) = S \cup (R \star \dom(\bH))$.
So, we are interested in studying   those sets of positions
$X$
such that  $ X= S \cup (R \star X)$.
To this aim, we define

\squishlist
\item[] {\centering
$ A_0 \ \eqdef \ S \enspace, \ \ \quad A_{n+1} \ \eqdef \  R \star A_n\enspace, \ \ \quad A_{< n} \ \eqdef \ \bigcup_{k < n}A_k  \enspace \  $ and $ \ \quad A \ \eqdef \  \bigcup_{n \in \NN} A_n \enspace.$ \par}
 \squishend

\noindent The key  property, that we show in  a moment, is that $A$ is the \emph{unique} set
 such that $ A= S \cup (R \star A)$.
This fact is known as \emph{Arden's Rule}
in the literature of formal languages
theory (see \eg \cite{manna}).

  We now  prove this fact in our setting (our proof is adapted from the one given in \cite{manna}).

\squishlist
\item[$\phantom{ab}$ (i)] \emph{Arden's Rule}: \ \ \  $ X= S \cup (R \star X)$ \ \  if and only if \ \
$X= A$.
 \squishend
 \noindent\emph{Proof of} (i).
 If  $X = A$, then   $A =  \bigcup_{n \in \NN} A_n = A_0 \cup \bigcup_{n >0}A_n = A_0 \cup (R \star \bigcup_{n \in \NN}A_n) =  S \cup (R \star A)$.

 To show the converse, let $X$ be such that $ X= S \cup (R \star X)$.
 For each $n \in \NN$  we define
$ X_0 \eqdef X \ , \ X_{n+1} \eqdef R \star X_n$.
By induction on $n$, we show  that for every  $n \in \NN$ we have
$X = A_{< n} \cup X_n$.
If $n=0$, then $  X = \emptyset \cup X_0 = X$.
Let  $n = m+1$ and assume $A_{< m} \cup  X_{m} =X$.
We have  $ A_{< m+1}   \cup X_{m+1} =
 A_0 \cup  (R \star  A_{< m})  \cup (R \star X_{m})
 = A_0 \cup (R \star (A_{< m} \cup  X_{m})) = S \cup (R \star X) = X
$.
Now,
since $ \roo \notin R$,  for  each $p \in X$   we have $p \notin X_{\len(p)+1}$, as  $\len(q) > \len(p)$
for every
$ q \in X_{\len(p)+1}$.
Hence, if $p \in  X =  A_{< \len(p)+1} \cup X_{\len(p)+1}$,
then   $p \in A_{< \len(p)+1} \subseteq A$.
Finally,
if $p \in A$, then $p \in  A_{n}$ for some $n \in\NN$. Therefore, $p \in A_n \subseteq
 A_{< n+1} \cup X_{n+1} = X$.
\qed

By (i),   the domain of \emph{any} solution of the recursive
equation $ \bbv \eqrec \bF$ has to be $A$.
However, this fact does not give us any hint about  how  to define  the \emph{values}
of a solution
$\bG$ (\ie $\bG(p)$, for $p\in A$).

To tackle this problem, we now  show the following  properties.

\squishlist
\item[$\phantom{ab}$ (ii)]  (a) For every $n\in \NN$, each  $ p \in A_n$
 can be  factorized in a \emph{unique} way as \linebreak $\phantom{ijjjjbj} p =  r^p_{0} \star \cdots \star  r^p_{n-1}\star s^p$,
where
 $r^p_{0} ,\ldots , r^p_{n-1} \in  R $  and  $ s^p \in S$.

$\phantom{ab}$\!\! (b) \, \, \,\!\!\! For every $n \in \NN$ and every $m \in \NN$, if $n \neq m$
then $A_n $ and $A_m$ are disjoint.

 \squishend
 \noindent\emph{Proof of} (ii).
(a) By induction on $n$.
If $n = 0$, then $ p \in A_0 = S$.
 By Lemma
 \ref{lemsub}(a),  for no $r \in R$ and
 $t \in \NN^\star$ we have $p = r \star t$. So, we set  $s^p \eqdef  p$.  If $n = m+1$,
then  $p \in A_{m+1} = R \star A_m$.
By Lemma \ref{lemsub}(b),  there is a  unique $r \in R$ and  a unique  $q \in A_m$
such that $p = r \star q$.
By inductive hypothesis, $q$ can be   written  as
$q=  r^q_{0} \star \cdots \star r^q_{m-1}\star s^q$.
Hence,    $p$ can be   factorized as $p=  r^p_{0} \star \cdots \star r^p_{n-1}\star s^p$, where   $r^p_0 \eqdef r$,
$r^p_{k} \eqdef r^q_{k-1}$ for every $ 1 \leq k < n$,
and $s^p \eqdef s^q$.

(b) Suppose, for a contradiction,  that  for some $p \in A$  there are
 $ n$ and $m$  in $ \NN$ such that $ n < m$,
 $ p \in A_{n}$ and $p \in A_{m}$. Let $n$ be the least natural number for which this fact holds.
 By Lemma \ref{lemsub}(a), we have $n > 0$, as for no $t \in \NN^\star$ it is the case that
$ p \in A_0 = S$ and $ p = r \star t \in A_m$. Suppose now that $n> 0$,
$ p = r \star t \in A_{n}$ and $ p = r' \star t' \in A_m$. Since $p \in R \star A$, we can apply  Lemma \ref{lemsub}(b) and we obtain
$r= r'$ and $t=t'$.
But then, we have
 $ t \in A_{n-1}$ and $t \in A_{m-1}$.
This contradicts the minimality of $n$.
 \qed

 By (ii), it follows that each $p \in A$ can be uniquely factorized as  $ p =  r^p_{0} \star \cdots \star  r^p_{n_p-1}\star s^p$, where
$n_p$ is the unique $ n \in \NN$ such that $ p \in A_n$.
We are now ready to define $\bG$.

 Let $p \in A$ and
let
 $m_p \in \NN$   be the cardinality
of  $\{  j < n_p \st \bF(r^p_j) = \neg \bbv\}$. We  set
$\dom(\bG) \eqdef A$ and

\squishlist
\item[] {\centering
 ${\bG}(p) \ \eqdef \ $ $\left\{
  \begin{array}{rl}
    \bF(s^p)\enspace &  \hbox{if $m_p$ is even} \\
    \neg\bF(s^p)\enspace &  \hbox{if $m_p$ is odd\enspace, \quad for  $p  \in \dom(\bG)$\enspace.}   \end{array}
\right.$\par}
\squishend

We now show that $\bG$ is a formula and that it is a solution of $\bbv \eqrec \bF$.

\squishlist
\item[$\phantom{ab}$ (iii)]

(1) \  $\roo \in \dom(\bG)$ and if    $q \sqsubseteq p$, then $ q \in \dom(\bG)$.

\hspace{0.49cm}(2) \ If  $\bG(p) \in \cV \cup \neg \cV$,  then  $p$ is a leaf of $\bG$.

\hspace{0.49cm}(3) \ $\bG = \bF[\bG/\bbv]$.
 \squishend
 \noindent\emph{Proof of} (iii).
(1)  $\roo \in S = A_0 \subseteq  \dom(\bG)$.
If $q  \sqsubseteq p$ then  either $q =
 r^p_0 \star \cdots \star r^p_{n_p-1}\star t$ and $ t \sqsubseteq s^p$,
or $q= r^p_0 \star \cdots  \star r^p_{k-1}\star t$, for $ k < n_p$
and  $t \sqsubset r^p_k$. In both cases, $ t \in S$. Thus,   $q \in A_{n_p} \cup  A_k \subseteq A = \dom(\bG)$.

(2) Let  $\bG(p) \in \cV \cup \neg \cV$. As $\bG(p)\in \{\bF(s^p) ,\neg \bF(s^p) \}$, we have $\bG(p)\in \cV \cup \neg \cV$ just in case $\bF(s^p)\in \cV \cup \neg \cV$.
Since    $s^p$ is a leaf of $\bF$ and $s^p \notin R$,  no proper extension of $p$ is in $\dom(\bG)$.
So, $p$ is a leaf of $\bG$.

(3)  By (i), $\dom(\bG) = A = S \cup (R \star A) = S \cup (R \star \dom(\bG))  =
\dom(\bF[\bG/\bbv]) $.
Let  $ p  \in \dom(\bG)$.
If $n_p=0$, then  $p = s^p \in S$ and $\bG(p) = \bF(p) = \bF[\bG/\bbv](p)$.
Suppose now that $ n_p>0$.
 Let  $q $ be $r^p_{1} \star \cdots \star r^p_{n_p-1}\star s^p$. By definition, $s^p = s^q$ (see the proof of (ii)(a)).
If  $\bF(r^p_0) = \bbv$, then $m_p = m_{q}$ and
 $\bG(p) =   \bG(q)= \bF[\bG/\bbv](p)$.
If  $\bF(r^p_0) =  \neg \bbv$, then $m_p = m_{q} +1$ and
 $\bG(p) =\neg \bG(q)= \bF[\bG/\bbv](p)$. \qed

 We now prove that $\bG$ is the unique solution of the recursive equation  $\bbv \eqrec \bF$.

\squishlist
\item[$\phantom{ab}$ (iv)] If  $\bH$ is a formula such that $\bH = \bF[\bH/\bbv]$, then $\bH = \bG$.
 \squishend
 \noindent\emph{Proof of} (iv).
By (i), $\dom(\bH) = A$.  We now prove that for every $n \in \NN$, for each $p \in A_n$ we have
 $\bH(p) = \bG(p)$.
We reason by induction on $n$.
If $ n = 0$, then $ p\in S$. Hence,
$\bH(p) = \bF[\bH/\bbv](p) = \bF(p) = \bG(p)$.
Suppose now that $ n = m+1$.
Let  $q $ be $r^p_{1} \star \cdots \star r^p_{n-1}\star s^p \in A_m$. By inductive hypothesis,  $\bH(q) = \bG(q)$.
  If  $\bF(r^p_0) = \bbv$, then $\bH(p) = \bF[\bH/\bbv](p)
= \bH(q) = \bG(q) = \bG(p)$.
 If  $\bF(r^p_0) = \neg \bbv$, then $\bH(p) = \bF[\bH/\bbv](p)
= \neg \bH(q) = \neg \bG(q) = \bG(p)$.
 \qed

Finally, we show that $\bG$ is ill--founded.

\squishlist
\item[$\phantom{ab}$ (v)] There exists a function $f : \NN \arr \dom(\bG)$ such that
$f(n) \sqsubset f(n+1)$, for every $n \in \NN$.
 \squishend

 \noindent\emph{Proof of} (v).
Recall that  $R$ is non--empty and that $\roo \notin R$. Let $r \in R$.
Define
\ $r^0 \eqdef \roo$  and $r^{n+1} \eqdef r \star r^n$, for each $n \in \NN$.
As $\roo \in S$,  we have  $r^n = r^n \star \roo \in A_n$
and $r^n \sqsubset r^{n+1}$, for every $n \in \NN$.
So, $\{r^n \st n \in \NN\} \subseteq A = \dom(\bG)$.
Thus, the function $f$ given by: $f(n) \eqdef r^n$ for every $n \in \NN$,
has the required property. \qed

The proof of Theorem \ref{recthem} is now complete.
\end{proof}

 \begin{Example} \label{A to A}
Let us consider the recursive equations
\squishlist
\item[] {\centering  (1) \ $\bu \ \eqrec \ \vee[\bu, \bu]$\enspace, \ \  \qquad
(2) \ $\bw \ \eqrec \  \wedge[\bw, \bw]$\quad \   and  \quad \  (3) \ $\bbv \ \eqrec \  \vee[\neg \bbv ,\bbv]$\enspace. \par}
\squishend

\noindent Let $\bU$, $\bW$ and $\bV$ be the solutions of the equations (1), (2) and (3) respectively.
We have  $\dom(\bU) = \dom(\bW) = \dom(\bV) =A$, where $A \eqdef \{ p \in \NN^\star \st p(k) \in \{ 0 ,1\}, $ for all $k <\len(p) \}$.
Moreover,  for $p \in A$, we have $\bU(p) = \vee$,
 $\mathbf{W}(p) = \wedge$ and
 \squishlist
\item[] {\centering
${\bV}(p) \ = \ $ $\left\{
  \begin{array}{rl}
    \vee\enspace&  \hbox{if  the cardinality of $\{ k \st p(k) = 0 \}$  is even} \\
    \wedge\enspace &  \hbox{if  the cardinality of $\{ k \st p(k) = 0 \}$  is  odd \enspace.} \\  \end{array}
\right.$ \par}
\squishend
In our setting,  $\bbv \eqrec \vee[\neg \bbv , \bbv]$ can be used to represent   the equation ``$X = X \to X$" which is a well--known  example of equation of \emph{mixed variance} in the literature of recursive types (see \eg \cite{MelVou}).
\hfill$\triangle$
 \end{Example}

\subsection{Derivations} \label{subder}
In this subsection, we introduce the notions of sequent and derivation. We also
show some  sequents which are derivable in our framework.
In this work, we   define sequents as special disjunctive formulas.  We think that this
choice  is convenient, as it  makes clearer the ``duality"
between sequents and environments (Definition \ref{envir}(2)).
Derivations are defined to be   trees labeled by sequents and rules. Our definition of derivation is
actually very similar to that of
 \emph{pre--proof} in \cite{PTLC} (very roughly, a pre--proof is a
 not necessarily well--founded
 derivation in
a sequent calculus for
classical logic with the
 $\omega$--rule).

 \begin{Definition}[Sequent] \label{seq}
We call   \textbf{sequent}   any  \emph{disjunctive} formula $\bF$ whose arity is $\{0,\ldots,n-1\}$, for some $n \in \NN$.
The set of all sequents is denoted by $\seq$.
  \hfill $\triangle$
\end{Definition}

\begin{Notation} \label{not}
 Let $\bS = \vee[\bF_0,\ldots, \bF_{n-1}]$ be a sequent.  Let  $\bG$ be a formula. We write
$\vee[\bF_0,\ldots, \bF_{n-1}, \bG]$,  and sometimes also $\bS \vee \bG$,
 for the sequent $\bT$ of arity
$\{0,\ldots,n\}$  defined as follows: $\dom(\bT) \eqdef \dom(\bS)
\cup \{ \vv n \ww \star q \st q \in \dom(\bG) \}$ and
\squishlist
\item[] {\centering

${\bT}(p) \ \eqdef \ $ $\left\{
  \begin{array}{rl}
    \bS(p)\enspace &  \hbox{if $p \in \dom(\bS)$} \\
    \bG(q)\enspace &  \hbox{if $p = \vv n \ww \star q$\enspace , \quad for  $p  \in \dom(\bT)$\enspace.}   \end{array}
\right.$
 \par}
\squishend
 In particular, if $n= 0$  then  $\bT = \FALSE \vee \bG =  \vee[\bG]$ (recall that $\FALSE$ is an abbreviation for $\vee[\ ]$).\hfill $\triangle$
\end{Notation}

 \begin{Definition}[Rule]
 A \textbf{rule}  is either an \emph{axiom
 rule} or a \emph{disjunctive rule} or a \emph{conjunctive rule}.
 \squishlist
 \item An \textbf{axiom rule} is an ordered triple $\Dis{\bbv, k,\ell}$, where   $\bbv \in \cV$, $k \in \NN$ and $\ell \in \NN$.
 \item
A \textbf{disjunctive rule} is an ordered triple  $\Dis{\vee, k,i_0}$, where $k \in \NN$ and $i_0 \in \NN$.
 \item A \textbf{conjunctive rule} is a ordered pair  $\Dis{\wedge, k}$, where $k\in \NN$.
  \squishend
  The set of all rules is denoted by $\rules$.
    \hfill $\triangle$
 \end{Definition}

 \begin{Definition}[Derivation] \label{deri}
 A \textbf{derivation}  is a tree $T$ labeled by
$\seq \times \rules$
such that
  for each $p \in \dom(T)$
 one of the following conditions  ($\axi$),  ($\vee$) and  ($\wedge$) holds.

\vspace{0.2cm}

{ \centering
\begin{tabular}{cc}
\!\!($\axi$) \!\! : \!\! $\left\{
  \begin{array}{ll}
 \mbox{\!\!(i)\!\!} &
  \hbox{\!\!$T(p) = \big( \vee[\bF_0,\ldots, \bF_{n-1}] \ , \ \Dis{\bbv, k,\ell} \big)$;} \\
 \mbox{\!\!(ii)\!\!}   &\hbox{\!\!$k< n$, $\ell < n$,    $\bF_k = \bbv$ and $\bF_\ell = \neg \bbv$;} \\
\mbox{\!\!(iii)\!\!} & \hbox{\!\!$p$ is a leaf of $T$.}
  \end{array}
\right.$
&
\!\!($\vee$) \!\! : \!\! $\left\{
  \begin{array}{ll}
 \mbox{\!\!(i)\!\!} &
  \hbox{\!\!$T(p) = \big(  \vee[\bF_0,\ldots, \bF_{n-1}] \ , \ \Dis{\vee, k,i_0} \big)$;} \\
 \mbox{\!\!(ii)\!\!}   &\hbox{\!\!$k < n$, $\bF_k = \arpa{\vv i \ww}{I}$ and $i_0 \in I$;} \\
\mbox{\!\!(iii)\!\!} & \hbox{\!\!$p \star \vv i \ww \in\dom(T)$ if and only if $i = i_0$;} \\
\mbox{\!\!(iv)\!\!} & \hbox{\!\!$T_\sL(p \star \vv i_0 \ww) = \ \vee[\bF_0,\ldots, \bF_{n-1},\bG_{\vv i_0 \ww}]$.}
  \end{array}
\right.$
\end{tabular}
\vspace{0.3cm}
\par}

{ \centering
\!\!($\wedge$) \!\! : \!\! $\left\{
  \begin{array}{ll}
 \mbox{\!\!(i)}\!\! &
  \hbox{\!\!$T(p) = \big(  \vee[\bF_0,\ldots, \bF_{n-1}] \ , \ \Dis{\wedge, k} \big)$;} \\
 \mbox{\!\!(ii)\!\!}   &\hbox{\!\!$k < n$ and $\bF_k = \arn{\vv i \ww}{I}$;} \\
\mbox{\!\!(iii)\!\!} & \hbox{\!\!$p \star \vv i \ww \in\dom(T)$ if and only if $i \in I$;} \\
\mbox{\!\!(iv)\!\!} & \hbox{\!\!$T_\sL(p \star \vv i \ww) = \ \vee[\bF_0,\ldots, \bF_{n-1}, \bG_{\vv i \ww}]$, for every $ i \in I$.}
  \end{array}
\right.$
\vspace{0.2cm}

\par}

\noindent  In the sequel, we use   $\pi, \rho,\sigma ,\ldots$ to range over derivations.
  We say that  \textbf{$\pi$ is a derivation of $\bS$} if $\pi_\sL(\roo) = \bS$, and we say
that  $\bS$  \textbf{is derivable} if there exists
a derivation $\pi$ of it.
\hfill$\triangle$
\end{Definition}

 \begin{remark} \label{remsubform}  Let  $\pi$ be a derivation.
Let $\pi_\sL(\roo) =   \vee[\bG_{0},\ldots, \bG_{m-1}]$,
 and  let $\pi_\sL(p) =  \vee[\bF_{0}, \ldots, \bF_{n-1}]$,
for some $p \in \dom(\pi)$.

(a) \emph{The subformula property.} \ \!
For every   $k < n$, there exists $\ell < m$ such that  $\bF_k$ is a subformula
 of   $\bG_\ell$.

(b) \emph{Leaves.} \ \! The position $p$ is a leaf of $\pi$
if and only if  $p$ is as in  ($\axi$)
of Definition \ref{deri}, or
$p$ is as in  ($\wedge$)
of Definition \ref{deri} and
$\bF_k = \TRUE$ (recall that $\TRUE$ is an abbreviation for $\wedge[\ ]$).

(c) \emph{Rules.} \ \! The rule $\pi_\sR(p)$ and  the sequent $\pi_\sL(p)$ completely determine the sequent $\pi_\sL(q)$,
for each $q$ which immediately extends
$p$ in $\dom(\pi)$. In other words, if $\rho$ and $\sigma$ are two derivations of the same sequent $\bS$,  then
 $\dom(\rho) = \dom(\sigma)$  and $\rho_\sR(p) = \sigma_\sR(p)$ for all $p \in \dom(\rho)$ together imply   $\rho = \sigma$. \hfill$\triangle$
 \end{remark}
Using  a more traditional notation for  derivations in the sequent calculus,  conditions ($\axi$),   ($\vee$) and ($\wedge$)
of Definition \ref{deri} can be respectively  written as follows:

\vspace{0.25cm}

{\small

{\centering
\begin{tabular}{ccc}

\AxiomC{}
\RightLabel{($\axi$)}
\UnaryInfC{$\vdash \ \ \Gamma \, , \,  \bF \, , \, \Delta \, , \, \neg \bF \, , \, \Sigma$}
\DisplayProof

&

  \def\fCenter{ \vdash \ \ }
\Axiom$\fCenter  \Gamma \, , \, \arpp{\vv i \ww}{I} \, , \, \Delta \, , \, \bF_{\vv i_0 \ww}$
\RightLabel{($\vee$)}
\UnaryInf$\fCenter  \Gamma \, , \, \arpp{\vv i \ww}{I} \, , \, \Delta$
\DisplayProof

&

  \def\fCenter{ \vdash \ \ }
\Axiom$\fCenter  \Gamma \, , \, \arnn{\vv i \ww}{I} \, , \, \Delta \, , \, \bF_{\vv i \ww} \quad  \ldots \mbox{ for all } i \in I$
\RightLabel{($\wedge$)}
\UnaryInf$\fCenter  \Gamma \, , \, \arnn{\vv i \ww}{I} \, , \, \Delta$
\DisplayProof

\end{tabular}
\vspace{0.25cm}

\par}

}

\noindent where  $\{ \bF ,  \neg \bF\}
= \{\bbv , \neg \bbv\}$ in  ($\axi$), $i_0 \in I$  in  ($\vee$), and  in  ($\wedge$) the expression
``$\ldots$ for all $i \in I$" means that
 there is one premise for each $i \in I$ .
In particular, if $I = \emptyset$, then
there is no premise above the conclusion.

The \emph{rules of inference} displayed above essentially correspond
to the \emph{normal rules} of Tait--calculus \cite{Tait68}. But in contrast with \cite{Tait68},  in our setting ill--founded derivations
  are permitted. As a consequence, we have the following results.

 \begin{theorem} \label{propder} \label{recder} \hfill
 \sqi
 \item[\emph{(1)}]
 Let \/$\bS =  \vee[\bF_{0}, \ldots, \bF_{n-1}]$\@ be a sequent. Suppose that for some  \/$k < n$\@  either
 \/$\bF_k$\@ is a disjunctive formula whose arity is not the empty set,  or  \/$\bF_k$\@ is a conjunctive formula.
Then,  \/$\bS$\@ is derivable.
\item[\emph{(2)}] Let  \/$\bbv \eqrec \bF$\@ be a recursive equation, and
let \/$\bG$\@ be its  unique solution.
Then, the sequent
\/$\vee[\bG]$\@ is
derivable.
\sqe
 \end{theorem}
 \begin{proof} (1)
Let $\bF_k = \vee_I \bG_{\vv i \ww}$, for some
$I \neq\emptyset$. Let
$ i_0\in I$. We define a tree $T$ labeled by
$\seq \times \rules$ as \linebreak

\vspace{-0.45cm}
\noindent follows:
$\dom(T) \eqdef \{ p \in \NN^\star \st p(k)= i_0, $ for all $k <\len(p) \}$ and

\squishlist
\item[] {\centering
${T}(p) \ \eqdef \ $ $\left\{
  \begin{array}{cl}
   \big( \bS \ , \ \Dis{\vee, k , i_0} \big)\enspace&   \hbox{if $p = \roo$} \\
   \big(    T_\sL(q) \vee  \bG_{\vv i_0\ww} \ , \ \Dis{\vee, k , i_0} \big)\enspace &  \hbox{if $p = q \star \vv i_0\ww$\enspace, \qquad for  $ p \in \dom(T)$\enspace.}    \end{array}
\right.$ \par}
\squishend
\noindent
Suppose now that $\bF_k = \wedge_I \bG_{\vv i \ww}$.  We analogously define a tree $U$ labeled by
$\seq \times \rules$ as follows: \linebreak

 \vspace{-0.45cm}
\noindent
$\dom(U) \eqdef \{ p \in \NN^\star \st p(k) \in I, $ for all $k <\len(p)\}$ and

\squishlist
\item[] {\centering
${U}(p) \ \eqdef \ $ $\left\{
  \begin{array}{cl}
   \big(  \bS \ , \ \Dis{\wedge, k} \big)\enspace&   \hbox{if $p = \roo$} \\
   \big(    U_\sL(q) \vee  \bG_{\vv i \ww} \ , \ \Dis{\wedge, k}  \big)\enspace &  \hbox{if $p = q \star \vv i \ww$\enspace, \qquad for  $ p \in \dom(U)$\enspace.}    \end{array}
\right.$ \par}
\squishend
Since each $p \in \dom(T)$ (resp. $p \in \dom(U)$)
satisfies condition
 $(\vee)$ (resp. $(\wedge)$) of Definition \ref{deri},
$T$  (resp. $U$)   is a derivation of $ \bS$.

(2) By condition (R$_2$) of Definition \ref{receq}, either $\bF = \vee_I \bF_{\vv i \ww}$ or $\bF = \wedge_I \bF_{\vv i \ww}$. Furthermore, by
 (R$_3$) we have that
$I \neq\emptyset$, as $\bbv$ or $\neg\bbv$ (or both)
must occur in $\bF$.
For $\Diamond \in \{\vee, \wedge\}$,  we have, by  Proposition \ref{subst}(2), that
  $ \bG =  \bF[\bG/\bbv] = (\Diamond_I \bF_{\vv i \ww})[\bG/\bbv] = \Diamond_{I}(\bF_{\vv i \ww}[\bG/\bbv]) = \Diamond_I \bG_{\vv i \ww}$. Hence, $\bG$ is a compound formula whose arity is $I \neq \emptyset$.
Then, we can apply (1) above to $\bS = \vee[\bG]$
and obtain the desired result.
     \end{proof}

\begin{remark} \label{remcons} Let  $\bU$, $\mathbf{W}$, and $\bV$ be the formulas defined   in Example \ref{A to A}.

(a) By  Theorem \ref{recder}(2),  for each $\mathbf{M}\in \{ \bU, \mathbf{W},\bV\}$
the sequent $\vee[\mathbf{M}]$
 is derivable.
But, since the  subformulas of $\mathbf{M}$  are neither atoms nor equal to $\TRUE$,       it follows from Remark \ref{remsubform}(a)(b) that
every derivation of  $\vee[\mathbf{M}]$ has to be ill--founded.

(b)  There are several sequents which are not derivable. For instance,
 $\FALSE$, $\vee[\FALSE], \vee[\FALSE,\FALSE]$, \ldots. Also, for each atom $\bba$  the sequents  $\vee[\bba]$,
 $\vee[\bba,\bba]$, $\vee[\bba,\bba,\bba]$,
 \ldots are not derivable.

(c) Even though $\FALSE$ is not derivable, our system is \emph{inconsistent}
in the sense that there is some formula $\bF$ such that
both  $\vee[\bF]$ and $\vee[\neg\bF]$ are derivable.
For, consider  $\bU$ and $\mathbf{W}$, and note that $\bU = \neg\mathbf{W}$. By (a) above, $ \vee[\bU]$ and
 $ \vee[\bW]$
are derivable.

(d) The \emph{cut--rule} is \emph{not admissible} in our
system, \ie it is not true that
\squishlist
\item[$\phantom{ab}$ (CUT)] \ for every sequent $\bS$ and every formula $\bF$,
if $\bS \vee \bF$
and $\bS \vee \neg \bF$ are  derivable, then so is
$\bS$.
\squishend
For, let $\bS$ be $\FALSE$.
By (c) above, both   $\FALSE\vee \bU = \vee[\bU]$ and
and $\FALSE\vee \neg \bU = \FALSE\vee \mathbf{W}   = \vee[\mathbf{W} ]$
are derivable. But by (b) above, $\FALSE$ is not derivable.
 \hfill$\triangle$
\end{remark}

\section{Interactive Semantics} \label{sec4}

We now present our interactive semantics
for derivations in infinitary classical logic.

Let $\bS$ be a sequent.
Recall that in this paper
we are not interested
in studying the concept of derivability, \ie  ``$\bS$ is derivable." Rather,
the concept that we want to  analyze is    ``the test $\cT$ comes from a derivation
of \/$\bS$", in the sense we now make clear.
Before giving the formal definitions, we explain at an informal  level the
 three basic notions of our semantics:
\emph{test}, \emph{environment} and \emph{interaction tree}. \\

\vspace{-3.5mm}

\noindent \textbf{Tests}. \  In this paper,  a test is nothing but
 a tree $\cT$ labeled by $\rules$ such that
$\dom(\cT) = \NN^\star$. This is actually the official definition of test  (that we repeat, for convenience, in the next subsection). Tests have to be thought as ``derivation candidates" for a derivation $\pi$  of a sequent
$\bS$. More precisely,  let the \emph{skeleton of the derivation $\pi$ of $\bS$} be
the tree $T$ labeled by $\rules$ whose domain is equal to $\dom(\pi)$ and such that $T(p) = \pi_\sR(p)$, for all $p \in \dom(\pi)$.  In other words,
the skeleton is just the object obtained
from a derivation  by erasing all the sequent information.
Then,
we say that \emph{the test $\cT$ comes from the derivation $\pi$ of $\bS$}
(or that $\cT$ is a ``successful candidate")
if   the skeleton of $\pi$ is ``contained"
into $\cT$ (see Definition \ref{interpretation} for the precise definition).
In general,  it is not  always the case that a test
contains the skeleton of some derivation.
In this paper, tests have to be understood as  \emph{untyped} objects in the sense that they are not necessarily related to formulas and sequents:
it may be the case that a test comes from a single derivation, or several derivations, or no derivation at all. The semantics we  define gives us a precise answer to the problem of determining when a test comes from a derivation or not.

 Finally, we also point out that the definition of ``to come from"
is \emph{syntactical}, as we use
the syntactical notions of derivation and skeleton  to define it.\\

\vspace{-3.5mm}

\noindent \textbf{Environments and interaction trees}. \
Technically speaking, an environment is just
the negation $\neg \bS$ of a sequent $\bS$ (in particular, it is not a sequent).  In our setting,  tests and environments
\emph{interact}.
To be more specific, a pair
$\big(\cT,\neg \bS\big)$, called
\emph{configuration},
uniquely determines a tree which is labeled by configurations of the previous kind, or by  an error symbol $\Uparrow$. We call this tree the \emph{interaction tree of} $\big(\cT,\neg \bS\big)$.
The construction of the interaction tree is \emph{dynamical} in the sense that it can be seen as the (possibly
infinitary) ``unfolding" of the transition relation |  of a suitable  transition system |  starting from  $\big(\cT,\neg \bS\big)$. The procedure which determines the interaction tree has also a very simple and
natural  \emph{game interpretation}.
Now, recall that the statement of the \emph{interactive completeness theorem} is:
\vspace{-0.00cm}
\squishlist
\item[] {\centering
$\cT$ comes from a derivation  of $\bS$ \ \ \ \ \ \ if and only if \ \ \ \ \ \ $\left.
  \begin{array}{l}
  \hbox{the interaction between $\cT$ and $\neg \bS$} \\
  \hbox{does not
produce errors.}   \end{array}
\right.$ \par}
\squishend

\noindent We can now intuitively explain
how errors
come into play:
if for some position $p$ in the domain of the interaction tree of $\big(\cT,\neg \bS\big)$ it is the case that the label is
$\Uparrow$, then we  say that \emph{the interaction between the test $\cT$ and the environment $\neg \bS$ produces an error}.
 Observe that the notion of ``to produce an error" is
\emph{semantical} in the sense that we use the interaction
tree to determine the eventual presence of errors. In particular, we make no use of  the syntactical notions of derivation and skeleton. \\



\vspace{-0.3cm}

We now develop the technical part of this section.
In Subsection \ref{sec trees} we give the formal definitions
of test, environment and interaction tree. We also give
a game theoretical interpretation  of our   interaction trees.
 In Subsection \ref{compsec} we  prove the interactive completeness theorem.

\subsection{Interaction trees} \label{sec trees}

\begin{Definition}[Test, environment, configuration]
\label{envir} \hfill
\sqi
\item[(1)] A  \textbf{test} is a tree  $T$ labeled by $\rules$ such that $\dom(T) = \NN^\star$. The set of all tests is denoted by $\TES$, and
we use  $\cT, \cU, \ccW,\ldots$ to denote its elements.
\item[(2)] We call \textbf{environment}   any formula
which is the negation of a sequent. That is, an environment
is a  \emph{conjunctive} formula
whose arity is $\{0,\ldots,m-1\}$, for some
$m \in \NN$.
The set of  all environments is denoted by $\envir$.
\item[(3)] We define the set $\conf$ of  \textbf{configurations} as
\squishlist
\item[]{ \centering
$\conf \ \eqdef \  \big(  \TES \!  \ \times \! \ \envir \big) \ \cup \ \{ \Uparrow\}$,
\par}
\squishend
\noindent where $\Uparrow$ (pronounced ``error")    is any object such that  $\Uparrow \ \notin \TES \times \envir$.
In the sequel, we use the letters $c, d, e , \ldots $ to denote configurations.\hfill$\triangle$
\sqe

\end{Definition}

\begin{Notation}
Let $\bE  = \wedge[\bG_0,\ldots, \bG_{m-1}]$
 be an environment.
Let    $\bH$ be a formula. In the sequel, we write
$\wedge[\bG_0,\ldots, \bG_{m-1}, \bH]$,   and when convenient also   $\bE \wedge \bH$,
 for the environment
$ \neg \big(\vee[\neg \bG_0,\ldots, \neg \bG_{m-1}, \neg\bH]\big)$  of arity $\{0, \ldots , m\}$. Recall that the sequent $\vee[\neg \bG_0,\ldots, \neg \bG_{m-1}, \neg\bH]$ is defined in
Notation \ref{not}.
 \hfill $\triangle$

\end{Notation}

 \begin{Definition}[Interaction tree, production of errors] \label{comptree}
 Let $c$ be a configuration.
The  \textbf{interaction tree of $c$}
is the tree $T$ labeled by $\conf$
 defined as follows.

\squishlist
\item[$\phantom{ab}$ (C$_1$)]  $\roo \in \dom(T)$,  and $T(\roo) \eqdef c$.

\item[$\phantom{ab}$ (C$_2$)]  Suppose that $p \in \dom(T)$ has been
 defined. We define the immediate  extensions of
 $p$ in $\dom(T)$  \linebreak    $\phantom{ab i}$  and their labels by cases as follows:

\squishlist

\item[$\phantom{ab}$ ($\Uparrow_1$)] \hspace{-0.11cm}If $T(p) = \ \Uparrow$, then $p$ is  a leaf of $T$.

\item[$\phantom{ab}$ ($\axi$)] \hspace{-0.07cm}If  $T(p) =  \big( \cT  \, , \, \wedge[\bG_0,\ldots, \bG_{m-1}]    \big)$,
  $\cT(\roo) = \Dis{\bbv, k,\ell}$,  $k <m$, $\ell <m$, $\bG_k = \neg \bbv$ and $\bG_\ell =  \bbv$, then $p$ \\
$\phantom{lii}$ is a leaf of $T$.

\item[$\phantom{ab}$ ($\vee$)] If  $T(p) =  \big( \cT  \, , \, \wedge[\bG_0,\ldots, \bG_{m-1}]    \big)$,
$\cT(\roo) = \Dis{\vee, k,i_0}$, $k < m$, $\bG_k = \arnnn{\vv i \ww}{I}$ and $ i_0 \in I$, then \linebreak
$\phantom{iii}$ $ p \star \vv i \ww \in \dom(T)$ if and only if $i = i_0$,
and $T(p \star \vv i_0 \ww) \eqdef \ \big( \cT_{\vv i_0 \ww}  \, , \, \wedge[\bG_0,\ldots, \bG_{m-1}, \bH_{\vv i_0 \ww}] \big)$.

\item[$\phantom{ab}$ ($\wedge$)] If   $T(p) =  \big( \cT  \, , \, \wedge[\bG_0,\ldots, \bG_{m-1}]    \big)$,
$\cT(\roo)= \Dis{\wedge, k}$,  $k < m$ and $\bG_k = \arppp{\vv i \ww}{I}$, then  $ p \star \vv i \ww \in \dom(T)$ \linebreak
$\phantom{lii}$ if and only if $i \in I$,
and $T(p \star \vv i \ww) \eqdef \ \big( \cT_{\vv i \ww}  \, , \, \wedge[\bG_0,\ldots, \bG_{m-1}, \bH_{\vv i \ww}] \big)$  for each $i \in I$.
\item [$\phantom{ab}$ ($\Uparrow_2$)]  \hspace{-0.20cm} In all the other cases, $ p \star \vv i \ww \in \dom(T)$ if and only if $i = 0$  and
 $T(p \star \vv 0 \ww) \eqdef \ \Uparrow$.
\squishend
\squishend
We denote the interaction tree of a configuration
$c$ as $\CT(c)$.

Let $\cT$ be a test,  and let $\bE$ be an environment.
We say that the \textbf{the interaction between $\cT$ and $\bE$ does not produce errors}
if
\squishlist
\item[] \centering{there is \emph{no} $p \in \dom(\CT(\big(\cT,\bE\big)))$
such that \,\! $\CT(\big(\cT,\bE\big))(p) = \ \Uparrow$ \enspace.

\vspace{-0.57cm}
\hfill$\triangle$ \par}
\squishend
\end{Definition}

We now discuss the previous bunch of definitions.\\

\vspace{-0.35cm}

\noindent \textbf{Global view of the interaction tree: interaction and transition systems}. \
The procedure which determines the interaction tree has to be read as follows. First, we have
 $\roo \in \dom(\CT(c))$, by (C$_1$).
From $\roo$ and its label,
we calculate, by using the clauses in (C$_2$), all the positions
of length one in  $\dom(\CT(c))$. Now,
we do the same thing for
each position
of length one
 in $\dom(\CT(c))$, and we  obtain
 all the positions
of length two in $\dom(\CT(c))$, and so on. This procedure might  not stop
and produce an ill--founded tree.

We call  $\CT(c)$   \emph{interaction} tree
because (in the case $c \neq \ \Uparrow$)
when a test and an environment are combined into a configuration, they \emph{interact}. Namely, they evolve into  new configurations according to
the clauses given in Definition \ref{comptree}.
 Furthermore, given a configuration $c$ its interaction tree $\CT(c)$ is defined
\emph{dynamically}: at each step one checks which clause
of Definition \ref{comptree} holds and then (possibly) proceeds to the
next step.  Indeed, one
can recognize the computational structure of \emph{transition system}  here: $ \conf$
is the set of   \emph{states} and
the \emph{transition relation} $\leadsto$ is
given according to the clauses
in (C$_2$).
Here,   the configuration  $\Uparrow$ represents a final state.
Thus, the interaction tree of $c$ can be seen
 as the (possibly
infinitary) ``unfolding" of the  transition relation $\leadsto$ starting from the state  $c$.\\

\vspace{-0.35cm}

\noindent \textbf{Local view of the interaction tree: a game theoretical interpretation}. \
At first view,   clauses (C$_2$)($\axi$),
(C$_2$)($\vee$) and (C$_2$)($\wedge$)
 of Definition \ref{comptree}
seem to be just  the ``dual"  of   conditions
($\axi$), ($\vee$) and ($\wedge$)
of
Definition \ref{deri} respectively.
This is admittedly true,  but notice that
in Definition \ref{deri} we use ($\axi$), ($\vee$) and ($\wedge$)  to \emph{define} the notion of derivation, whereas here
 we use (C$_2$)($\axi$),
(C$_2$)($\vee$) and (C$_2$)($\wedge$) to (possibly) \emph{find errors}. Furthermore,
there is nothing which   corresponds to  clauses
(C$_2$)($\Uparrow_1$) and (C$_2$)($\Uparrow_2$) in Definition \ref{deri}. On the other hand,
these clauses are  crucial to the aim of finding errors.
Hence, even though derivations and interaction trees
seem to be similar,
they have  different purposes.

We can now  give a very simple and natural game interpretation
to the objects introduced so far. Specifically,
the interaction tree can be seen as   a  \emph{debate} between two players:   $P$ (Proponent) and $O$ (Opponent). The moves of the game
are determined by   the clauses   given in Definition \ref{comptree}.
In our game,
$P$ is always in charge of tests,
whereas
$O$ is always in charge of environments.
Given a configuration $c$,
$P$ (possibly)  makes a question and $O$
answers to the question by  producing a set of configurations. Then, for each  configuration produced by  $O$,
$P$ (possibly) makes a question and $O$
answers to the question by  producing a set of configurations,  and so on.
The debate starts with a given configuration $c$ different from $\Uparrow$. Then, for each configuration $d$ given or  produced so far, the two players
 behave as follows.
\sqi
\item[$\phantom{ab}$($\Uparrow$G)] \ \  If $d = \ \Uparrow$, then  $P$ does not make any question and  $O$ does not produce any configuration.

\sqe
Let  $d =
\big( \cT  \, , \, \wedge[\bG_0,\ldots, \bG_{m-1}]    \big)$. Then, $P$ asks $\cT(\roo)$
and   $O$  answers as follows.
\sqi
\item[$\phantom{ab}$ ($\axi$G)]  If $P$   asks  $\Dis{\bbv, k,\ell}$, then
$O$  checks the formulas $\bG_k$ and  $\bG_\ell$.
If $\bG_k = \neg \bbv$ and $\bG_\ell =  \bbv$, then   $O$
does \linebreak  $\phantom{ijiiii}$ not
produce any configuration.
Otherwise,
$O$
produces $\Uparrow$.

\item[$\phantom{ab}$ ($\vee$G)]  If $P$   asks  $\Dis{\vee, k,i_0}$, then
$O$   checks the formula $\bG_k$.
If  $\bG_k = \arnnn{\vv i \ww}{I}$ and $ i_0 \in I$,
then  $O$
 produces \linebreak  $\phantom{ijij}$
$\big( \cT_{\vv i_0 \ww}  \, , \, \wedge[\bG_0,\ldots, \bG_{m-1}, \bH_{\vv i_0 \ww}] \big)$.
Otherwise,
$O$
produces $\Uparrow$.

\item[$\phantom{ab}$ ($\wedge$G)]  If $P$   asks  $\Dis{\wedge, k}$, then
$O$   checks the formula $\bG_k$.
If  $\bG_k = \arppp{\vv i \ww}{I}$,
then  $O$
 produces \linebreak  $\phantom{ijij}$
$\big( \cT_{\vv i \ww}  \, , \, \wedge[\bG_0,\ldots, \bG_{m-1}, \bH_{\vv i \ww}] \big)$, for each $i \in I$.
Otherwise,
$O$
produces $\Uparrow$.

  \sqe
Of course, up to terminology and the rearrangement of clause ($\Uparrow_2$)
inside the other clauses, this is  just a more informal reformulation  of the clauses given in  Definition \ref{comptree}.  We say that $O$ \emph{wins the debate starting
from $c$} if at some stage it is able to exhibit
an error, \ie to produce  the configuration
$\Uparrow$. Otherwise, we say that $P$ \emph{wins the debate starting
from $c$}.
With this game theoretical intuition in mind,
the statement of the interactive completeness theorem
can be reformulated as follows: $P$  wins
the debate starting
from $c = (\cT, \neg \bS)$  just in case
$\cT$ comes from a derivation  of $\bS$.
\\

\vspace{-0.35cm}

\noindent \textbf{Special cuts.} \ A configuration of the form
$\big( \cT  \, , \, \wedge[\bG_0,\ldots, \bG_{m-1}]    \big)$ can also be understood as a \emph{special cut}, \ie as a  cut
  of the following special form
(here we use a more traditional notation)

\vspace{-1.5mm}

{\small

{\centering
\AxiomC{$ \vdots \pi$}
\noLine
\UnaryInfC{$\vdash  \bF_0,\ldots,  \bF_{m-1}$}
\AxiomC{$ \vdots z_{\bG_0}$}
\noLine
\UnaryInfC{$\vdash \bG_0$}

\AxiomC{\!\!\!\!\!\!$\ldots$\!\!\!\!\!\!}

\AxiomC{$ \vdots z_{\bG_{m-1}}$}
\noLine
\UnaryInfC{$\vdash  \bG_{m-1}$}
\RightLabel{cut}
\QuaternaryInfC{}
\DisplayProof
\par} }

\vspace{2.0mm}

\noindent where we simultaneously  cut $\bF_k$ with $\bG_k$
for each $k < m$. Here, the derivation $\pi$ (in the sense of the previous section)  corresponds to the tests $\cT$ and, for each formula  $\bG$ of our logic, the derivation $z_\bG$ is  given by:
 if $\bG$ is an atom, then  $\vdash \bG$ has no premises; if $\bG = \arpp{\vv i \ww}{I}$ or
$\bG =  \arnn{\vv i \ww}{I}$,  then there is one premise  $ \vdash \bF_{\vv i \ww}$  for each   $i \in I$
(thus, $z_\bG$ is not a derivation in the sense of the previous section, and it is akin to a derivation of the subformula tree of $\bG$).
In addition,
 the formulas  $\bF_k$ and $\bG_k$ need not be the negation of  each other.
Hence, \emph{special cuts
are not cuts in the ordinary sense}, and  in particular,
special cuts have nothing to do with the notion of cut discussed in the previous section.

Under this interpretation, an environment $\wedge[\bG_0,\ldots, \bG_{m-1}] $ can be naturally understood as a \emph{conjunction} of formulas, as it represents a sequence $\vdash \bG_0 \, \ldots \,  \vdash \bG_{m-1}$ of (unary) traditional sequents.
With this picture in mind, the clauses given in
 Definition \ref{comptree} transform a  special cut
 into a set of special cuts or $\Uparrow$. So, they work as  steps of reduction for a procedure of cut--elimination
 for special cuts.
For instance,  in the case of  (C$_2$)($\vee$)
we  obtain after one step of reduction (here $\bF_k = \vee_{J} \bL_{\vv i \ww}$ is any disjunctive formula such that   $i_0 \in J$)

\vspace{-1.5mm}

{\small
{\centering
\AxiomC{$ \vdots \pi_{\vv i_0 \ww}$}
\noLine
\UnaryInfC{$\vdash  \bF_0,\ldots, \vee_{J} \bL_{\vv i \ww} , \ldots, \bF_{m-1},\bL_{\vv i_0\ww}$}
\AxiomC{$\!\!\!\!\!\!\!\!\!\!\! \vdots z_{\bG_0}$}
\noLine
\UnaryInfC{$\vdash \bG_0$ \quad $\ldots$}

\AxiomC{$\!\!\!\!\!\!\!\!\!\!\! \vdots z_{\bG_k}$}
\noLine
\UnaryInfC{$\vdash \wedge_{I} \bH_{\vv i \ww}$\quad $\ldots$}

\AxiomC{$ \vdots z_{\bG_{m-1}}$}
\noLine
\UnaryInfC{$\vdash \bG_{m-1}$}

\AxiomC{$ \vdots z_{\bH_{\vv i_0 \ww}}$}
\noLine
\UnaryInfC{$\vdash \bH_{\vv i_0 \ww}$}
\RightLabel{cut}
\QuinaryInfC{}

\DisplayProof \par} }

\vspace{0.25cm}

\noindent which is again a
special cut.
The present discussion on special cuts is quite
informal, but the  point here is to develop
another intuition about the  concepts  introduced so far.

\subsection{Interactive completeness} \label{compsec}
In this section, we prove the interactive completeness theorem. First,
we give the formal definition of ``to come from",
and then we move to the statement and the proof of the theorem.

 \begin{Definition}[To come from] \label{interpretation} Let $\bS$ be a sequent, and let $\pi$ be a derivation of $\bS$.
Let $\cT$ be a test. We say that $\cT$
\textbf{comes from the derivation $\pi$ of $\bS$}
if \squishlist
\item[] \centering{$\cT(p) \ = \  \pi_\sR(p)$\enspace,   \ for every $p \in \dom(\pi)$\enspace.

\vspace{-0.57cm}
\hfill$\triangle$ \par}

\squishend

 \end{Definition}

In other words, $\cT$ comes from the derivation $\pi$ of $\bS$ if $\cT$ ``contains" the \emph{skeleton} of $\pi$, in the  sense of the beginning of this section.

 \begin{theorem}[Interactive completeness]
 Let  \/$\bS$\@ be a sequent, and let  \/$\cT$\@ be a test.
 Then, the following claims are equivalent.
 \squishlist
 \item[$\phantom{ab}$ \emph{(1)}] The test \/$\cT$\@ comes from a derivation  of  \/$\bS$\@.
 \item[$\phantom{ab}$ \emph{(2)}] The interaction between \/$\cT$\@ and  \/$\neg \bS$\@
 does not produce errors.
 \squishend
  \end{theorem}

 \begin{proof} \ (1) implies (2) \ \!\! :   \ \!\!
 Let $\pi$  be a derivation of $\bS$, and assume that $\cT$ comes from $\pi$. Let $c$ be $\big( \cT  \, , \, \neg\bS    \big)$.  To show that the interaction between \/$\cT$\@ and  \/$\neg \bS$\@
 does not produce errors,
 we   now prove that
 every
$p \in \NN^\star$ satisfies
one of the following (mutually exclusive) conditions:
\squishlist
\item[$\phantom{ab}$ (P$_1$)] : \ $p \in  \dom(\pi)$,
$p \in  \dom(\CT(c))$
and $\CT(c)(p) =  \big( \cT_p  \, , \, \neg \pi_\sL(p) \big)$.
\item[$\phantom{ab}$ (P$_2$)] : \   $p \notin \dom(\pi)$ and
$p \notin  \dom(\CT(c))$.
\squishend
 We proceed  by induction on the length of $p$.
The position $\roo$ satisfies (P$_1$),
as $\roo \in \dom(\pi)$, $\roo \in \dom(\CT(c))$  and $\CT(c)(\roo) = c =  \big( \cT  \, , \, \neg \bS \big) = \big( \cT_{\roo}  \, , \, \neg \pi_\sL(\roo) \big)$.
Consider now an  arbitrary position $p \in \NN^\star$.
Assume that $p$ satisfies (P$_1$) or (P$_2$).
We  show that for each $j \in \NN$, the position
$p\star \vv j \ww$ satisfies (P$_1$) or (P$_2$).
If $p$ satisfies (P$_2$), then
$p \star \vv j \ww \notin \dom(\pi)$ and
$p \star\vv j \ww \notin  \dom(\CT(c))$, as $\pi$ and $\CT(c)$
 are labeled trees. Hence,
 $p \star \vv j \ww$ satisfies (P$_2$).
Otherwise,  $p$ satisfies   (P$_1$). Since
 $p \in \dom(\pi)$, we have  $\cT(p) = \pi_\sR(p)$, as $\cT$ comes from $\pi$ by assumption.
 Moreover,
since $\pi$ is a derivation, $p$  satisfies
one of
 conditions  ($\axi$), ($\vee$) and ($\wedge$) of Definition
\ref{deri}. Let
$\pi_\sL(p)$ be  $\vee[\bG_{0},\ldots, \bG_{n-1}]$,
so that $\neg \pi_\sL(p)= \wedge[\neg\bG_{0},\ldots, \neg \bG_{n-1}]$.
 We now consider   the following subcases.

(i) If  $\cT(p) = \Dis{\bbv , k , \ell}$,
then $p$ satisfies
 ($\axi$) of Definition
\ref{deri}.
 Hence, $k<n$, $\ell < n$,
 $ \bG_k = \bbv$,
 $\bG_\ell = \neg \bbv$ and
  $p$ is a leaf of $\pi$. Since  $\neg \bG_k = \neg \bbv$ and
 $\neg \bG_\ell = \bbv$,   $\CT(c)(p)$
is as in (C$_2$)($\axi$)
 of Definition \ref{comptree}.  Hence,
$p$ is  a  leaf of
  $\CT(c)$. In particular,
$p \star \vv j \ww$ satisfies (P$_2$).

(ii) If $\cT(p) = \Dis{\vee , k , i_0}$, then  $p$
satisfies   ($\vee$) of Definition
\ref{deri}. This means that $ k < n$,
 $\bG_k = \arppp{\vv i \ww}{I}$, $i_0 \in I$,
 only $p \star \vv i_0 \ww$ immediately extends $ p$ in $\dom(\pi)$,
 and $\pi_\sL(p \star \vv i_0 \ww) =
  \pi_\sL(p) \vee \bH_{\vv i_0 \ww}$.
Since   $\neg \bG_k = \wedge_{I} \neg\bH_{\vv i \ww}$ and $i_0 \in I$,     $\CT(c)(p)$
is as in (C$_2$)($\vee$)
 of Definition \ref{comptree}. Hence, only $p \star \vv i_0 \ww$ immediately extends $ p \in \dom(\CT(c))$
 and $\CT(c)(p \star \vv i_0 \ww) = \big( (\cT_p)_{\vv i_0 \ww} \, , \,
 \neg \pi_\sL(p) \wedge \neg \bH_{\vv i_0 \ww}
   \big) =
   \big( \cT_{p\star \vv i_0 \ww} \, , \,
 \neg\pi_\sL(p \star \vv i_0 \ww)
   \big)
   $. So, $p \star \vv  j\ww$ satisfies (P$_1$) if $ j = i_0$, and  $p \star \vv  j\ww$ satisfies (P$_2$) otherwise.

 (iii) Finally, if  $\cT(p) = \Dis{\wedge , k}$,
then $p$ satisfies
 ($\wedge$) of Definition
\ref{deri}. So,
 $ k < n$,
 $ \bG_k = \arnnn{\vv i \ww}{I}$,
$p \star \vv i \ww$ immediately extends $ p$ in $\dom(\pi)$ if and only if $i \in I$,
 and $\pi_\sL(p \star \vv i \ww) =
  \pi_\sL(p) \vee \bH_{\vv i \ww}$
  for each $i \in I$.
Since   $\neg \bG_k = \vee_{I} \neg\bH_{\vv i \ww}$,  $\CT(c)(p)$
is as in   (C$_2$)($\wedge$)
 of Definition \ref{comptree}.  Thus,
$p \star \vv i \ww$ immediately extends $ p$ in  $\dom(\CT(c))$ if and only if
$i \in I$
 and $\CT(c)(p \star \vv i \ww) = \big( (\cT_p)_{\vv i \ww} \, , \,
 \neg   \pi_\sL(p) \wedge \neg \bH_{\vv i \ww}
   \big) =
   \big( \cT_{p\star \vv i \ww} \, , \,
 \neg\pi_\sL(p \star \vv i \ww)
   \big)
   $, for every $i \in I$.
Hence, $p \star \vv  j\ww$ satisfies (P$_1$) if $ j \in I $, and  $p \star \vv  j\ww$ satisfies (P$_2$) otherwise.

This proves that  each $p \in \NN^\star$ satisfies
either  (P$_1$) or (P$_2$).
From  this fact, it follows that
$\CT(c)(p) =
\big( \cT_p  \, , \, \neg \pi_\sL(p) \big)$ for all $p \in\dom(\CT(c))$.
In particular, there  is no $p \in \dom(\CT(c))$
such that $\CT(c)(p) = \ \Uparrow$, \ie the interaction between \/$\cT$\@ and  \/$\neg \bS$\@
 does not produce errors.

\vspace{0.05cm}

  (2) implies (1)   :      Let $c$ be $\big( \cT  \, , \, \neg\bS    \big)$ and assume that (2) holds.
This means that each position $p \in \dom(\CT(c))$
 is labeled by a member of $\TES \times \envir$.  In  such a situation, it immediately  follows from the definition of interaction tree that     for each
 $p \in \dom(\CT(c))$ we have
$\CT(c)(p) = \big( \cT_p \, , \, \bE    \big)$, for some
 $\bE \in \envir$.
We  define a tree $T$ labeled $\seq \times \rules$ as follows: $\dom(T)  \eqdef   \dom(\CT(c))$ and  \squishlist
\item[] {\centering
 $T(p) \ \eqdef \     \big(\neg \bE \, , \, \cT(p)\big)$\enspace, \ \
 for  $ p  \in \dom(T)$  with   $\CT(c)(p) = \big(\cT_p \, , \, \bE\big)$\enspace.
 \par}

\squishend
We now show that $T$ is a derivation. To do this,
we need to check that for each $p \in \dom(T)$
one of    conditions  ($\axi$),  ($\vee$) and  ($\wedge$) of Definition
\ref{deri} holds.
Let  $\CT(c)(p) = \big( \cT_p \, , \, \bE \big)$ and let $\bE$ be $\wedge[\bF_{0},\ldots,\bF_{m-1}]$, so that
$\neg \bE = \vee[\neg\bF_{0},\ldots,\neg\bF_{m-1}]$. There are the following cases to consider.

(1) Suppose that $\cT(p) = \Dis{\bbv , k , \ell}$.
Since the interaction between \/$\cT$\@ and  \/$\neg \bS$\@
 does not produce errors, the configuration
$\CT(c)(p)$
is as  in (C$_2$)($\axi$)
 of Definition \ref{comptree}. This means that
 $k < m$, $\ell < m$,
 $ \bF_k = \neg \bbv$,
 $ \bF_\ell =  \bbv$ and
  $p$ is a leaf of $\CT(c)$.  Since $\neg\bF_k = \bbv$ and
 $\neg\bF_\ell = \neg \bbv$,  and since  $p$ is  a leaf of $T$, the position
 $p$ satisfies   ($\axi$) of Definition
\ref{deri}.

(2) Suppose that $\cT(p) = \Dis{\vee , k , i_0}$. Since the interaction between \/$\cT$\@ and  \/$\neg \bS$\@
 does not produce errors,  the situation is as in
 (C$_2$)($\vee$)
 of Definition \ref{comptree}. Hence,
 $k < m$,
 $ \bF_k = \wedge_{I} \bH_{\vv i \ww}$, $i_0 \in I$,
 only $p \star \vv i_0 \ww$ immediately extends $ p$ in $\dom(\CT(c))$,
 and   $\CT(c)(p \star \vv i_0 \ww) = \big( \cT_{p \star \vv i_0 \ww}   \, , \, \bE \wedge \bH_{\vv i_0 \ww} \big)$.
  Since  $\neg \bF_k = \vee_{I} \neg \bH_{\vv i \ww}$,
  $i_0 \in I$,
 only $p \star \vv i_0 \ww$ immediately extends $ p$ in $\dom(T)$
 and   $T_\sL(p \star \vv i_0 \ww) = \neg(\bE \wedge \bH_{\vv i_0 \ww}) = \neg \bE \vee \neg \bH_{\vv i_0 \ww}$, we conclude that
 ($\vee$) of Definition
\ref{deri} holds for the position $p$.

(3) Finally, suppose that $\cT(p) = \Dis{\wedge, k}$.
Since the interaction between \/$\cT$\@ and  \/$\neg \bS$\@
 does not produce errors, the configuration
$\CT(c)(p)$ is as in  (C$_2$)($\wedge$)
 of Definition \ref{comptree}. Thus,
 $k < m$,
 $\bF_k = \vee_{I} \bH_{\vv i \ww}$,
$p \star \vv i \ww$ immediately extends $ p$ in $\dom(\CT(c))$ if and only $i \in I$,
 and   $\CT(c)(p \star \vv i \ww) = \big( \cT_{p \star \vv i \ww}   \, , \, \bE \wedge \bH_{\vv i \ww} \big)$
 for every $i \in I$.
  Since $\neg \bF_k = \wedge_{I} \neg \bH_{\vv i \ww}$,
$p \star \vv i \ww$ immediately extends $ p$ in $\dom(T)$ if only if $ i \in I$,
   $T_\sL(p \star \vv i \ww) = \neg(\bE \wedge \bH_{\vv i \ww}) = \neg \bE \vee \neg \bH_{\vv i \ww}$ for every $i \in I$, we conclude that
 the position
 $p$ satisfies   ($\wedge$) of Definition
\ref{deri}.

Hence, $T$ is a derivation of $\bS$.
By construction, $T_\sR(p) = \cT(p)$ for all
$p \in \dom(T)$. Therefore, the test \/$\cT$\@ comes from a derivation of \/$\bS$\@, namely \/$T$\@. \end{proof}

\nocite{*}
\bibliographystyle{eptcs}
\bibliography{basaldellabibliography}

\end{document}